\documentclass[journal]{IEEEtran}
\usepackage[utf8]{inputenc}
\usepackage[T1]{fontenc}
\usepackage{graphicx,amssymb,amsmath,color,textcomp,bm,cite,stfloats,array,booktabs}
\usepackage[hidelinks]{hyperref}
\allowdisplaybreaks[1]

\newtheorem{remark}{\textbf{Remark}}
\newtheorem{proposition}{\textbf{Proposition}}
\newtheorem{corollary}{\textbf{Corollary}}

\begin{document}

\title{ICL-SEC: Iterative Cross-Layer Semantic \\ Error Correction

\thanks{Y. Wang, S. C. Liew, and Y. Du are with the Department of Information Engineering, The Chinese University of Hong Kong, Hong Kong SAR, China (e-mail:\{yrwang, soung, yuydu\}@ie.cuhk.edu.hk).} 
\thanks{Corresponding author: Soung Chang Liew.}

\author{Yirun Wang, Soung Chang Liew, {\it Life Fellow, IEEE}, and Yuyang Du, {\it Member, IEEE}}
}
\maketitle

\begin{abstract}
Iterative decoding has been central to the success of modern channel coding, where reliability information is repeatedly exchanged across decoding components to approach fundamental performance limits. This paper brings the same principle to semantic error correction by proposing \textit{iterative cross-layer semantic error correction} (ICL-SEC), a framework that closes the loop between the physical-layer soft channel decoder and the application-layer language-model-empowered semantic decoder. In the proposed framework, a soft-input soft-output channel decoder first produces bit-level posterior probabilities, from which word-level reliabilities are derived. Words deemed reliable are exposed to a masked language model as semantic context, while unreliable words are masked. The language model then produces contextual word likelihoods, which are leveraged to generate extrinsic bit-level priors and fed back to the channel decoder for the next iteration. This iterative refinement progressively expands the set of confidently recovered words. A key contribution is our \textit{Confirm} prior-update rule: once a word is judged reliable, its bits are assigned deterministic priors with probability one in subsequent iterations, making the word fully resolved side information for both the channel decoder and the language model. This successive-confirmation mechanism prevents oscillatory unmask-mask behavior and yields a reliability interpretation consistent across layers. 
Notably, ICL-SEC is agnostic to the underlying channel code, requiring only a soft-input soft-output decoder. We demonstrate this generality with convolutional and LDPC codes, using belief-propagation (BP) decoding in both cases -- BCJR on the trellis for the former and standard BP on the Tanner graph for the latter.
Simulations over text transmission demonstrate that ICL-SEC substantially outperforms both conventional channel decoding and non-iterative CL-SEC. In particular, the proposed Confirm scheme reduces the bit error rate by more than two orders of magnitude relative to non-iterative CL-SEC, while also significantly improving the other five performance metrics. Consistent gains carry over to the LDPC instantiation, where ICL-SEC reduces the bit error rate by more than four orders of magnitude relative to conventional BP decoding. These results show that principled iterative cross-layer semantic feedback can dramatically improve reliable and meaning-preserving communication.
\end{abstract}

\begin{IEEEkeywords}
Cross-layer error correction, iteration, language models, semantic communications.
\end{IEEEkeywords}

\section{Introduction}\label{introduction}

\IEEEPARstart{T}{he} triumphs of iterative decoding in traditional communication systems
demonstrate the transformative power of exchanging reliability
information across decoding components. When applied to
capacity-approaching codes such as turbo codes and low-density
parity-check (LDPC) codes, iterative decoding has enabled performance
close to the Shannon limit \cite{Ryan2009Channel}. Inspired by
these successes, this paper investigates whether a similar iterative
principle can be tailored to semantic error correction \textbf{(SEC)},
particularly to \textbf{cross-layer SEC} methods
\cite{Wang2026CLSEC,Li2026LLMViterbi,Hao2026Semantic,Yue2026Semantic}
that integrate physical-layer channel evidence with application-layer
semantic information. Our results indicate that, with a principled design, \textbf{iterative cross-layer SEC} can reduce the bit error rate by more than \textbf{two orders of magnitude} relative to non-iterative cross-layer SEC decoding for text transmission.

To motivate the proposed iterative design, we briefly review how SEC has
evolved from post-decoding correction to cross-layer joint decoding.
Although the idea of combining channel decoding with semantic inference
may also be applicable to non-textual sources, such as images, this
paper follows most existing SEC works and focuses on text recovery. In
this setting, SEC has recently emerged as an
application-layer complement to conventional forward error correction
(FEC). In a conventional SEC pipeline, the physical-layer decoder first
makes hard decisions on the received bits, and SEC is then applied as a
post-processing step to correct the resulting corrupted text using
language-model (LM) context \cite{Hao2025Shorta}.

However, this separation is generally suboptimal: once the decoder
output is reduced to hard decisions, much of the physical-layer
reliability information is discarded, leaving SEC without knowledge of
which bits or words are more reliable. In our earlier work
\cite{Wang2026CLSEC}, we proposed
\textbf{CL-SEC}, a \textbf{cross-layer joint decoding} framework that
fuses physical-layer soft information with LM-derived
semantic probabilities, thereby anchoring semantic correction to channel
evidence and markedly improving information recovery over conventional
\textbf{post-decoding SEC}.

Existing cross-layer SEC methods, however, typically perform this fusion
in a one-shot manner: physical-layer soft information is integrated once
with contextual likelihoods produced by LM-driven semantic inference to
form a cross-layer decision on the recovered text \cite{Wang2026CLSEC,Li2026LLMViterbi,Hao2026Semantic,Yue2026Semantic}. The
resulting semantic information is not fed back to refine the channel
decoder, and newly refined channel-decoder outputs are not repeatedly
used to update the semantic priors. This one-shot design leaves the
central promise of iterative cross-layer refinement largely unexploited.

This limitation echoes a classical lesson from channel coding: the power
of turbo and LDPC decoding lies not merely in combining distinct sources
of reliability information, but in exchanging and refining that
information iteratively across decoding components
\cite{Ryan2009Channel}. Motivated by this principle, we propose
\textbf{iterative CL-SEC (ICL-SEC)}, which adapts this idea to semantic
error correction by \textbf{closing the loop} between the physical-layer
\textbf{channel decoder} and \textbf{application-layer LM-driven
semantic decoder.}

\textbf{Overview of the ICL-SEC Framework:} The following steps outline the proposed ICL-SEC framework for information recovery at the receiver. Although the framework can be extended to non-textual content,
this work focuses on text transmission. 

\textbf{Step 1:} The channel decoder computes bit-level soft
information, from which word-level posterior probabilities are derived.
Words whose posterior probabilities exceed a prescribed threshold are
deemed reliable and left \textbf{unmasked}; the remaining words are
\textbf{masked}.

\textbf{Step 2:} The partially masked text sequence, e.g., ``Mary
\textless mask\textgreater{} a little \textless mask\textgreater,'' is
provided to an LM. Conditioned on the unmasked words, the LM produces
\textbf{contextual probability distributions over candidate words} at
each word position. 

\textbf{Step 3:} The LM-derived word probabilities are leveraged to
compute \textbf{extrinsic bit-level priors} and fed back to the channel
decoder. With these updated priors, the channel decoder repeats Step 1,
enabling iterative cross-layer refinement.

Across iterations, the objective is to progressively increase the number
of words recovered with high confidence.

A central design question in this iterative loop is \textbf{how to
update the bit-level priors} of words that have already been deemed
reliable. We consider three strategies. The \textbf{Naive} strategy
updates the priors of both unmasked and masked words at every iteration.
The \textbf{Hold} strategy freezes the current priors of an unmasked
word once it has been deemed reliable, so that only masked words
continue to receive updated priors. The proposed \textbf{Confirm}
strategy goes further: once a word is deemed reliable, the word is confirmed and its corresponding
bits are assigned deterministic priors with probability one in all
subsequent iterations. In this way, a confirmed word is treated as fully
resolved side information rather than as a soft and revisable estimate.

This above progression also reveals why the prior-update rule is central
to ICL-SEC. The naive strategy exhibits unsatisfactory performance and
oscillatory behavior, since words that are unmasked in one iteration may
be perturbed and masked again later. The hold strategy mitigates but
does not eliminate this instability, because held soft priors still do
not make an unmasked word fully reliable from the channel decoder's
perspective. These observations motivate the Confirm strategy, which
achieves the best performance among the three and provides a consistent
interpretation of reliability across the channel decoder and the LM. 
Rather than a pure belief-propagation (BP) algorithm as in turbo and LDPC decoding, this \textbf{successive-confirmation
mechanism} is more akin to \textbf{successive interference
cancellation} \textbf{(SIC)} decoding in wireless networks (see Remark~\ref{re:analogy} in
Section~\ref{confirm-clp-update-ultimate-proposed-scheme}).

For a comprehensive assessment, we evaluate ICL-SEC using six
metrics spanning transmission reliability and semantic
fidelity. Experimental results show that all three iterative
variants -- Naive, Hold, and Confirm -- outperform non-iterative CL-SEC
across the metrics, demonstrating the broad benefit of
iterative cross-layer refinement. Among these variants, the Confirm
method performs best, reducing both the bit error rate (BER) and word error rate (WER) by more than two
orders of magnitude relative to non-iterative CL-SEC. These results show
that the gains of ICL-SEC come not merely from repeated semantic
feedback, but from a principled prior-update rule that stabilizes the
interaction between channel decoding and semantic decoding.

Crucially, our framework requires only an arbitrary soft-input soft-output channel decoder. To demonstrate this versatility, we supplement our primary convolutional-code evaluation with an LDPC instantiation in Appendix~\ref{app:ldpc}. Under this alternative setup, the proposed Confirm scheme reduces the BER by more than four orders of magnitude relative to conventional BP decoding.

\section{Preliminaries}\label{preliminaries}

\subsection{System Architecture}\label{system-architecture}

Fig.~\ref{fig:system} shows the block diagram of the overall system. 
At the transmitter, we remove all spaces and punctuation from the
original text message. Metadata specifying the lengths of successive
words is then included in the header \cite{Wang2026CLSEC}. After word recovery at the
receiver, spaces are restored using this metadata, while punctuation is
restored using a context-aware modern LM \cite{Wang2026CLSEC}.

The remaining payload consists of
$N$
words, denoted by
$\mathbf{w}_{ n } ,n\in\mathcal{N}\triangleq \{1,\ldots,N\}$.
Let
$L_{ n }$
be the number of letters in
$\mathbf{w}_{ n }$.
The word-length information
$\{L_{ n } \}_{ n\in\mathcal{N} }$
is embedded in the frame header so that the receiver can determine the
word boundaries.\footnote{The word-length metadata can be efficiently
  compressed, e.g., using the canonical Huffman coding scheme \cite{Wang2026CLSEC}.
  We note that removing spaces and punctuation reduces the payload
  length more than the compressed header increases it, and thus the
  ICL-SEC frame can be shorter than the original full-text frame. The
  saved length can be used to provide stronger FEC protection for the
  header.} Concatenating the words yields the character sequence
$(\mathbf{w}_{ 1 } ,\ldots,\mathbf{w}_{ N } )$,
which is encoded using the ASCII code, yielding the information bit
stream
$\mathbf{v}=(v_{ 1 } ,\ldots,v_{ K } )$.
Thus,
$K_{ n } =8L_{ n }$
is the number of bits associated with word
$\mathbf{w}_{ n }$.

Before channel encoding, an interleaver is applied to the information
bit stream to disperse the bits of each word across the overall source-bit sequence so that the successive source bits presented to the
following channel encoder are not from the same word or from the
neighborhood of the word. In particular, the interleaver here plays a
role analogous to the interleaver in the turbo code. With the
interleaver,
$\mathbf{v}$
is permuted into
$\mathbf{u}\triangleq (u_{ 1 } ,\ldots,u_{ K } )=\pi(\mathbf{v}),$
where
$u_{\pi(k)}=v_k,$ $k\in\mathcal{K}\triangleq \{1,\ldots,K\}$.
We can obtain
$\mathbf{v}$
from
$\mathbf{u}$
via the inverse deinterleaving
$\mathbf{v}=\pi ^ { -1 } (\mathbf{u}).$

In general, ICL-SEC can be instantiated with any channel code that supports soft-input soft-output decoding, including but not limited to the convolutional code and the LDPC code~\cite{Ryan2009Channel}. As a representative instance, the main body of this paper adopts the convolutional code with the BCJR decoding -- a realization of BP algorithm on the trellis; this choice imposes no additional restriction on the proposed framework. Further \textbf{adaptation} of ICL-SEC to the \textbf{LDPC code} with the Tanner-graph-based BP decoding is detailed in Appendix~\ref{app:ldpc}.

In the convolutional-code setup, the bit stream $\mathbf{u}$ is encoded by a rate-$R$ convolutional encoder. For each information bit $u_{ k }$, the encoder produces $M\triangleq 1/R$ coded bits, denoted by $\mathbf{c}_k=(c_{k,1},\ldots,c_{k,M})$, where $c_{k,m}\in\mathbb{F}_2\triangleq\{0,1\}$ for $ m=1,\ldots,M $.

The coded bits are binary phase-shift keying (BPSK) modulated as
$\mathbf{x}_k\triangleq(x_{k,1},\ldots,x_{k,M})=2\mathbf{c}_k-\mathbf{1}_M\in\{-1,+1\}^{M},$ $ k\in\mathcal{K} $. 
The received symbols over the additive white Gaussian noise (AWGN)
channel are
$\mathbf{y}_k\triangleq(y_{k,1},\ldots,y_{k,M})=\mathbf{x}_k+\mathbf{n}_k,$
where the noises 
$\mathbf{n}_k\sim\mathcal{N}(\mathbf{0}_M,\sigma^2\mathbf{I}_M)$.
With the normalization of signal amplitude, the signal-to-noise ratio
(SNR) is given by
$\operatorname{SNR}= 1 / \sigma ^ { 2 } $.
For notational convenience, we define
$\mathbf{y}=(\mathbf{y}_{1},\ldots,\mathbf{y}_{K}).$

At the receiver, \textbf{ICL-SEC} is performed to recover the
transmitted word sequence; the detailed decoding procedure is deferred
to Section~\ref{icl-sec}. As illustrated in Fig.~\ref{fig:system}, ICL-SEC forms an iterative loop
between \textbf{BCJR decoding} and \textbf{LM-driven cross-layer prior (CLP)
update}. After each BCJR pass, bit-level posterior probabilities are
used to identify reliable and unreliable words. Reliable words are
retained, whereas unreliable words are masked and passed to the LM,
which produces contextual likelihoods for candidate words. These
contextual likelihoods are then combined with BCJR soft information by
the \textbf{CLP update} module to produce
updated bit-level priors, which are fed back to the BCJR decoder for the
next iteration.

\begin{figure*}[!t]
\centering
\includegraphics[width=0.8\linewidth]{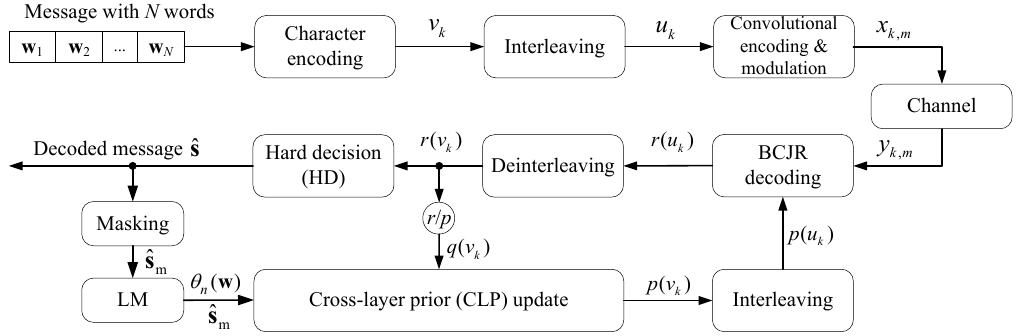}
\caption{Block diagram of the overall system, shown for an instantiation with a convolutional code decoded by BCJR algorithm (see Appendix~\ref{app:ldpc} for an LDPC-based instantiation). The ICL-SEC operates at the receiver for message recovery.}
\label{fig:system}
\end{figure*}

\subsection{BCJR Decoding}\label{bcjr-decoding}

We first review the BCJR decoding to facilitate a precise description of
ICL-SEC in Section~\ref{icl-sec}. The BCJR algorithm computes the \emph{a
posteriori} probability (APP) of each information bit by marginalizing
over all trellis branches associated with that bit
\cite{Ryan2009Channel}. At input bit
$k\in\mathcal{K}$,
let
$s_{ k } \in\mathcal{S}$
be the encoder state, and
$u_{ k } \in\mathbb{F}_{ 2 }$
be the input bit. Denote the constraint length of the convolutional code by
$C$,
and thus
$\lvert\mathcal{S}\rvert=2 ^ { C-1 }$, where $ C-1 $ is the memory length.
For a convolutional code, each trellis branch is specified by the
previous state
$s_{ k-1 } =s'$,
the input bit
$u_{ k } =u$,
and the current state
$s_{ k } =s$;
any two of these variables determine the remaining one. The APP
$\mathbb{P}(u_k=u | \mathbf{y})$
of
$u_{ k }$
can be expressed as
\begin{equation}
\mathbb{P}(u_k=u |\mathbf{y})\!\propto \!\! \sum_{(s',s)\in\mathcal{S}_u}
\!\!\!\mathbb{P}
\big(s_{k-1}=s',s_k=s,\mathbf{y}\big),\, u\in\mathbb{F}_2,
\label{eq:1}
\end{equation}
where
$\mathcal{S}_{ u }$
denotes the set of possible state transitions
$(s',s)$
associated with the input bit
$u_{ k } =u$.

Define the forward metric, the branch metric, and the backward metric as
\begin{align}
\alpha_{ k } \left ( { s } \right )&=\mathbb{P}\left ( { s_{ k } =s,\mathbf{y}_{ 1 } ,\ldots,\mathbf{y}_{ k } } \right ),
\label{eq:2}
\\
\gamma_k(s',s)&=\mathbb{P}(s_k=s,\mathbf{y}_k | s_{k-1}=s'),
\label{eq:3}
\\
\beta_k(s)&=\mathbb{P}(\mathbf{y}_{k+1},\ldots,\mathbf{y}_K | s_k=s).
\label{eq:4}
\end{align}
We then have the factorization:
\begin{equation}
\mathbb{P}(s_{k-1}=s',s_k=s,\mathbf{y})=\alpha_{k-1}(s')\gamma_k(s',s)\beta_k(s).
\label{eq:5}
\end{equation}
The forward metric
$\alpha_{ k } (s)$
is computed in a  recursive manner as
\begin{equation}
\alpha_{ k } \left ( { s } \right )=\sum \limits_{ s' } { \gamma_{ k } \left ( { s',s } \right )\alpha_{ k-1 } \left ( { s' } \right ) },
\label{eq:6}
\end{equation}
where the summation is over all possible encoder states
$s'$
that can transition to state
$s$.
Similarly, the backward metric
$\beta_{ k } (s)$
is computed recursively as
\begin{equation}
\beta_{ k-1 } \left ( { s' } \right )=\sum \limits_{ s } { \beta_{ k } \left ( { s } \right )\gamma_{ k } \left ( { s',s } \right ) }.
\label{eq:7}
\end{equation}
Finally, the branch metric
$\gamma_{ k } (s',s)$
can be written as
\begin{equation}
\gamma_k(s',s)=\mathbb{P}(u_k=u)\mathbb{P}(\mathbf{y}_k | u_k=u),\; (s',s)\in\mathcal{S}_u,
\label{eq:8}
\end{equation}
where
$\mathbb{P}\left ( { u_{ k } } \right )$
is the prior probability, and
$\mathbb{P}(\mathbf{y}_{k} | u_k)=\mathbb{P}(\mathbf{y}_{k} | \mathbf{x}_{k})$.
In conventional BCJR, the bit priors are set to
$\mathbb{P}(u_k=0)=\mathbb{P}(u_k=1)=1/2$.

\section{ICL-SEC}\label{icl-sec}

Fig.~\ref{fig:iclsec-detail} provides a detailed view of the ICL-SEC processing, serving as a
zoomed-in representation of the ICL-SEC component in the high-level
system diagram of Fig.~\ref{fig:system}. In Fig.~\ref{fig:iclsec-detail}, the upper part is a factor graph of
the convolutional code; the BCJR decoding as described in Section~\ref{bcjr-decoding} can be derived from this factor graph by applying the
belief-propagation (BP) rule. The masked language model (MLM) and CLP updater
at the lower part can be viewed as a \textbf{semantic decoder} separate from the
\textbf{BCJR decoder}.

\begin{figure}[!t]
\centering
\includegraphics[width=1.02\linewidth]{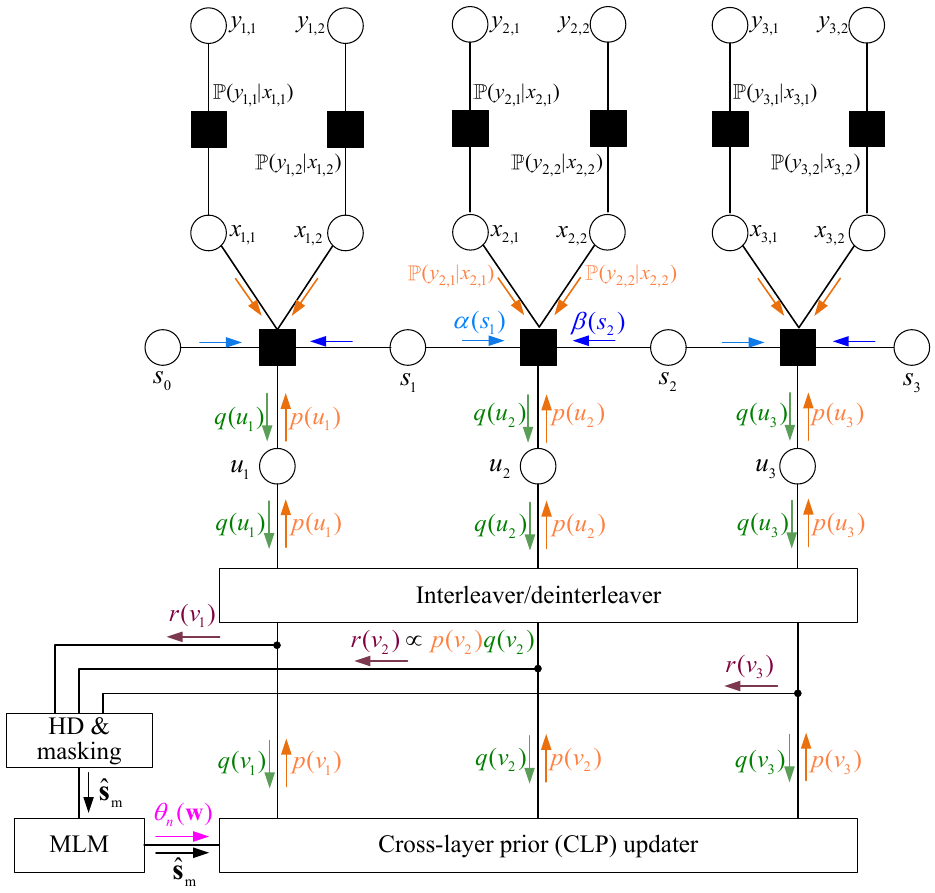}
\caption{ICL-SEC with $R=  1 / 2 $ and $K=3$. The $\gamma_{ k } (s',s)$ term in \eqref{eq:8} becomes $\gamma_k(s',s)=p(u_k)\mathbb{P}(y_{k,1} | x_{k,1})\mathbb{P}(y_{k,2} | x_{k,2}),$ where $(x_{k,1},x_{k,2})\in\{-1,1\}^{2}$ are the BPSK-modulated channel-coded symbols corresponding to the source bit $u_{ k } \in\mathbb{F}_{ 2 }$, and $(y_{k,1},y_{k,2})$ are the received BPSK symbols.}
\label{fig:iclsec-detail}
\end{figure}

\subsection{Overview}\label{overview}
Conventional BCJR assumes that the information bits are equally probable
and thus set
$\mathbb{P}(u_k=0)=\mathbb{P}(u_k=1)=1/2$ in
\eqref{eq:8}. However, text exhibits strong structural characteristics.
Consequently, an LM can exploit syntactic structures, lexical
regularities, and contextual dependencies to infer semantic information
about likely words. The semantic correlations among words uncovered by
LM can be leveraged to provide CLPs to BCJR for refined decoding. A
principle of ICL-SEC is to iteratively refine the values of
$\mathbb{P}(u_{ k } )$
in \eqref{eq:8} in successive BCJR decoding.

Specifically, the generated CLPs are fed into the BCJR decoder, whose
output is subsequently used by the LM-driven semantic decoder to update the CLPs. The refined
priors are then used in the next decoding iteration. The three major
steps, briefed introduced in Section~\ref{introduction}, are outlined below with
reference to Fig.~\ref{fig:iclsec-detail}. Precise details are provided in Part~B, Part~C, and Part~D.

Superscripts $i=1,\ldots,I$ in the following denote the iteration index.

\textbf{Initialization:} Set
$i=1$
and initialize the interleaved-order bit-prior $p^{(i)}(u_k=0)=p^{(i)}(u_k=1)=1/2$
for all
$k\in\mathcal{K}$.

\textbf{Step 1 -- BCJR Decoding and Masking (see Part~B):} Perform
BCJR decoding to produce the complete APP
$r ^ { (i) } (u_{ k } )$,
as given by \eqref{eq:1}. BCJR also generates the soft information
$q ^ { (i) } (u_{ k } )$
that will be fed into the CLP updater. Note that
$r ^ { (i) } (u_{ k } )\propto p ^ { (i) } (u_{ k } )q ^ { (i) } (u_{ k } )$,
where
$q ^ { (i) } (u_{ k } )$
is the intrinsic information and
$p ^ { (i) } (u_{ k } )$
is the extrinsic information from the perspective of the BCJR decoder.
After deinterleaving $\pi^{-1}$, we have
$r ^ { (i) } (v_{ k } )$,
from which word-level posterior probabilities are derived. The
hard-decision (HD) and masking module in Fig.~\ref{fig:iclsec-detail} examines the probability
of the most likely word at each word position. If the probability
exceeds a prescribed threshold
$T$,
the word is left unmasked; otherwise, it is replaced by a mask symbol.

\textbf{Step 2 -- MLM-Driven Contextual Inference (see Part~C):}
The sequence of unmasked words and masks, denoted by
$\hat {\mathbf{s}}_{ \mathrm{m} } ^{ (i) }$
in Fig.~\ref{fig:iclsec-detail}, is provided to a masked language model (MLM). Conditioned on
the visible words and their surrounding context, the MLM produces a
probability distribution over candidate words for each word position
\emph{n}, denoted by
$\theta_{ n } ^{ (i) }(\mathbf w)$
in Fig.~\ref{fig:iclsec-detail}.

\textbf{Step 3 -- CLP Update (see Part~D):} As shown in Fig.~\ref{fig:iclsec-detail},
$\hat {\mathbf{s}}_{ \mathrm{m} } ^{ (i) }$,
$\theta_{ n } ^{ (i) }(\mathbf w)$,
and
$q ^ { (i) } (v_{ k } )$
are provided as inputs to the CLP updater. Based on
$( \hat {\mathbf{s}}_{ \mathrm{m} } ^{ (i) },\theta_{ n } ^{ (i) }(\mathbf w),q ^ { (i) } (v_{ k } ))$,
the CLP updater computes the updated bit-level priors, denoted by
$p ^ { (i+1) } (v_{ k } )$.
The updated priors are fed back to the BCJR decoder as extrinsic
information in its next-iteration decoding.

The above procedure is repeated until the maximum iteration count is
reached or an early stopping rule is met. Specifically, early stopping
is triggered when no unconfirmed words remain after the HD-and-masking
step, or the change in the priors relative to the previous iteration
becomes sufficiently small. After the final iteration, the decoded
message
$\hat {\mathbf{s}}$
is constructed using the decoded words
$\{ \hat {\mathbf{w}}_{ n } ^{ (i) }\}_{ n\in\mathcal{N} }$
as the final output. Note that, after the final iteration, the
most-likely words at each position are left unmasked whether its
probability exceeds the threshold
$T$
or not, and there are no masks in the final sequence
$\hat {\mathbf{s}}$.

In the following, we elaborate on the procedure in
Fig.~\ref{fig:iclsec-detail}. Unless otherwise stated, we describe the bits
involved in the word-level operations, such as candidate construction,
HD, masking, and CLP update, in the natural bit-order
$v_{ k }$.
The BCJR decoder operates in the interleaved order
$u_{ k }$.
After each BCJR pass, $r ^ { (i) } (u_{ k } )$ and $q ^ { (i) } (u_{ k } )$ are deinterleaved to obtain
$r ^ { (i) } (v_{ k } )$
and
$q ^ { (i) } (v_{ k } )$.
After the CLP update,
$p ^ { (i+1) } (v_{ k } )$ is
interleaved back to
$p ^ { (i+1) } (u_{ k } )$
for the next BCJR pass.

\subsection{BCJR-Intrinsic Information, Hard Decision, and Masking}\label{bcjr-intrinsic-information-hard-decision-and-masking}

This part elaborates on the BCJR decoding, word-wise HD, and
masking mechanism in \textbf{Step 1} of Part~A. 
To ensure probabilities in the distribution summing to one, we define a normalization operator $\mathsf{N}_{ \mathbf{a},\mathcal{A} }$:
\begin{equation}
\mathsf{N}_{ \mathbf{a},\mathcal{A} } \left ( { f\left ( { \mathbf{a} } \right ) } \right )=\frac { f\left ( { \mathbf{a} } \right ) } { \sum \limits_{ \mathbf{b}\in\mathcal{A} } { f\left ( { \mathbf{b} } \right ) } }{ \rm{ , } }\;a\in\mathcal{A}.
\label{eq:10}
\end{equation}

\textbf{BCJR-Intrinsic Information:}
As shown in Fig.~\ref{fig:iclsec-detail}, given the input prior
$p ^ { (i) } (v_{ k } )$,
the BCJR decoder produces the complete BCJR APP
$r ^ { (i) } (v_{ k } )$,
given by
\begin{equation}
r ^ { (i) } \left ( { v_{ k } =v } \right )=\mathsf{N}_{ v,\mathbb{F}_{ 2 } } \Big ( { p ^ { (i) } \left ( { v_{ k } =v } \right )\cdot q ^ { (i) } \left ( { v_{ k } =v } \right ) } \Big ),\; v\in\mathbb{F}_{ 2 }.
\label{eq:9}
\end{equation}
The intrinsic information
$q ^ { (i) } (v_{ k } )$
(from the BCJR decoder's perspective) is output to the CLP updater as one
of its inputs.

\textbf{Word-Wise HD:}
The complete probabilities
$r ^ { (i) } (v_{ k } )$
are used to decide the hard word. Denote the complete vocabulary by
$\mathcal{W}$.
For each word
$\mathbf{w}_{ n } ,n\in\mathcal{N}$,
we construct a candidate set
$\mathcal{C}_{ n } \subset\mathcal{W}$
where each of the words has
$L_{ n }$
letters. Let
$\mathcal{B}_n=\{\tilde{K}_n+1,\ldots,\tilde{K}_{n+1}\}$
be the set of natural-order bit indices covered by word
$n$,
where
$\tilde{K}_n=\sum_{\ell=1}^{n-1}K_\ell$
with
$K_{ l } =8L_{ l }$.
For a candidate
$\mathbf{w}\in\mathcal{C}_{ n }$
at the word position
$n$,
let
$b_{ k } (\mathbf{w})\in\mathbb{F}_{ 2 }$
denote its ASCII bit at the natural-order bit-position
$k\in\mathcal{B}_{ n }$.

Note that a sequence of bits may or may not be a valid word
$\mathbf{w}\in\mathcal{C}_{ n }$.
To ensure that the probabilities of candidate words sum to one, we need
to normalize the raw bit-sequence probabilities of a word. 
For each word position
$n\in\mathcal{N}$,
we compute the probability of each valid-word candidate
$\mathbf{w}\in\mathcal{C}_{ n }$
from
$r ^ { (i) } (v_{ k } )$
by
\begin{equation}
\mu_n^{(i)}(\mathbf{w})=\mathsf{N}_{\mathbf{w},\mathcal{C}_n}\!\left(\prod_{\ell\in\mathcal{B}_n}r^{(i)}\big(v_\ell=b_\ell(\mathbf{w})\big)\right),\; \mathbf{w}\in\mathcal{C}_n.
\label{eq:11}
\end{equation}
Then we determine the most probable candidate and its confidence:
\begin{equation}
\hat {\mathbf{w}}_{ n } ^{ (i) }=\operatorname*{arg\,max} \limits_{ \mathbf{w}\in\mathcal{C}_{ n } } \, \mu_{ n } ^{ (i) }\left ( { \mathbf{w} } \right ), \;\; \Xi_{ n } ^{ (i) }=\max \limits_{ \mathbf{w}\in\mathcal{C}_{ n } } \mu_{ n } ^{ (i) }(\mathbf{w}).
\label{eq:12}
\end{equation}
Mapping
$( \hat {\mathbf{w}}_{ n } ^{ (i) })_{ n\in\mathcal{N} }$
back to bits yields the HD bit stream
$( \hat v_{ k } ^{ (i) })_{ k\in\mathcal{K} }$.

\textbf{Masking:}
For each
$n\in\mathcal{N}$,
if
$\Xi_{ n } ^{ (i) }\ge T$
for a threshold
$T\in[0,1]$,
then
$\hat {{\mathbf{w}}}_{ n } ^{ (i) }$
is left unmasked; otherwise,
$\hat {\mathbf{w}}_{ n } ^{ (i) }$
is replaced by a mask symbol (e.g., \textless mask\textgreater).

We define
$\mathcal{U} ^ { (i) } \subseteq \mathcal{N}$
as the set of word positions that are unmasked at iteration
$i=1,\ldots,I$
with the initialization
$\mathcal{U} ^ { (0) } =\emptyset$,
and define its complement
$\overline{\mathcal{U}}^{(i)}=\mathcal{N}\setminus\mathcal{U}^{(i)}$.
With these definitions, we form
\begin{equation}
\tilde{\mathbf{w}}_n^{(i)}=\begin{cases}\hat{\mathbf{w}}_n^{(i)},&n\in\mathcal{U}^{(i)},\\ \mathrm{\textless\!\! mask\!\!\textgreater},&n\in\overline{\mathcal{U}}^{(i)}.\end{cases}
\label{eq:13}
\end{equation}
Inserting spaces between adjacent words in
$\{\tilde{\mathbf{w}}_n^{(i)}\}_{n\in\mathcal{N}}$
yields the sequence of unmasked words and masks, denoted by
$\hat {\mathbf{s}}_{ \mathrm{m} } ^{ (i) }$.
An example of
$\hat {\mathbf{s}}_{ \mathrm{m} } ^{ (i) }$
is ``There \textless mask\textgreater{} a beach with palm
\textless mask\textgreater{} and clear blue water.''

\begin{remark}
The BCJR-intrinsic information
$q ^ { (i) } (v_{ k } )$
is output to the CLP updater as one
of its inputs. Note that, if the MLM and CLP updater in Fig.~\ref{fig:iclsec-detail} are
viewed as a semantic decoder separate from the BCJR decoder, then
$q ^ { (i) } (v_{ k } )$
and
$\hat {\mathbf{s}}_{ \mathrm{m} } ^{ (i) }$
can be interpreted as extrinsic information from the semantic decoder's
perspective. In particular, while
$q ^ { (i) } (v_{ k } )$
is soft information,
$\hat {\mathbf{s}}_{ \mathrm{m} } ^{ (i) }$
is hard information: a mask indicates no extrinsic information is given
at the associated word position, and an unmasked word conveys the
intention to fix the word without further changes (this is the case at least for the proposed Confirm scheme in Part~D). In that light, the threshold $T$ should be chosen close to
one when targeting a very low word error rate.

\end{remark}

\subsection{MLM-Driven Contextual Inference}\label{mlm-driven-contextual-inference}
This part elaborates on the MLM-driven contextual inference in
\textbf{Step 2} of Part~A. With reference to Fig.~\ref{fig:iclsec-detail}, the MLM (e.g.,
mmBERT \cite{Marone2025mmBERT}) takes the sequence of decoded words
and masks 
$\hat {\mathbf{s}}_{ \mathrm{m} } ^{ (i) }$,
as input, and produces a probability distribution over candidate words at
each word position. We restrict the MLM output to the length-compatible
candidate set
$\mathcal{C}_{ n } \in\mathcal{W}$.
After normalization over
$\mathcal{C}_{ n }$,
this gives the semantic word likelihood for each word $n\in\mathcal{N}$:
\begin{equation}
\theta_{ n } ^{ (i) }(\mathbf{w})=\mathsf{N}_{ \mathbf{w},\mathcal{C}_{ n } } \!\Big ( { \mathbb{P}\big(\mathbf{w} | \hat {\mathbf{s}}_{ \mathrm{m} } ^{ (i) },\mathrm{MLM}\big) } \Big ),\; \mathbf{w}\in\mathcal{C}_{ n } ,
\label{eq:14}
\end{equation}
where ``MLM'' in \eqref{eq:14} represents the adoption of a
specific instance of MLM. Both
$\hat {\mathbf{s}}_{ \mathrm{m} } ^{ (i) }$
and
$\theta_{ n } ^{ (i) }(\mathbf{w})$
are forwarded to the CLP Updater as its inputs.

\subsection{Confirm CLP Update (Ultimate Proposed Scheme)}\label{confirm-clp-update-ultimate-proposed-scheme}

This part elaborates on the CLP update in \textbf{Step 3} of
Part~A. In this paper, we study three CLP update schemes that share
the same components in Fig.~\ref{fig:iclsec-detail} except the CLP updater module. This
part focuses on the ``Confirm CLP updater'' (this paper's ultimate
proposed scheme). Part~E will elaborate on the other two schemes.

\vspace{1mm}
\noindent\textbf{1) Specification of Confirm CLP Updater}
\vspace{1mm}

As shown in Fig.~\ref{fig:iclsec-detail}, the CLP updater refines
$p ^ { (i) } (v_{ k } )$
using the masked text sequence
$\hat {\mathbf{s}}_{ \mathrm{m} } ^{ (i) }$,
MLM-produced word distribution
$\theta_{ n } ^{ (i) }(\mathbf{w})$,
and the intrinsic information of BCJR
$q ^ { (i) } (v_{ k } )$.
The Confirm CLP updater follows the principle of
\textbf{successive confirmation}. Once a word is unmasked, it is
confirmed and remain unmasked in all subsequent iterations. This is
achieved by setting
the bit priors of unmasked words to $0/1$ value (as elaborated below),
forcing the subsequent BCJR decoding to assign APP with probability one to the
unmasked words so that their probabilities will meet the threshold $ T $
 and these words remain unmasked in the next iteration.

\textbf{Update of}
$p ^ { (i) } (v_{ k } )$
\textbf{for Unmasked Words:}
Specifically, suppose that the word at position \emph{n} is unmasked
(i.e.,
$n\in\mathcal{U} ^ { (i) }$)
and
$(\hat{v}_k^{(i)})_{k\in\mathcal{B}_n}$
are the HD bits of the unmasked words. Then we set that
\begin{equation}
p^{(i+1)}(v_k=v)=\begin{cases}1,&v=\hat{v}_k^{(i)},\\0,&v\ne\hat{v}_k^{(i)},\end{cases}\;\; k\in\mathcal{B}_n,\ n\in\mathcal{U}^{(i)}.
\label{eq:15}
\end{equation}
This assignment in \eqref{eq:15} ensures consistent
reliability interpretation between the BCJR decoder and the MLM, as
elaborated in Remark~\ref{re:consisent}.

\textbf{Update of}
$p ^ { (i) } (v_{ k } )$
\textbf{for Masked Words:}
Consider a particular bit position \emph{k}. With reference to
Fig.~\ref{fig:iclsec-detail}, BCJR produces the APP
$r ^ { (i) } (v_{ k } )$.
This APP is a complete probability, since it already includes the
incoming prior
$p ^ { (i) } (v_{ k } )$,
as in \eqref{eq:9}. Therefore, the information passed from
the BCJR decoder to the CLP updater should remove the contribution of
this incoming prior to avoid double dips. Following this principle, we
leverage the BCJR intrinsic
$q ^ { (i) } (v_{ k } )$
rather than
$r ^ { (i) } (v_{ k } )$
to update the CLP.

Specifically, we combine
$q ^ { (i) } (v_{ k } )$
with the MLM likelihood
$\theta_{ n } ^{ (i) }(\mathbf{w})$
to generate the next bit prior for the masked words. For a target bit
$k\in\mathcal{B}_{ n }$
within a masked word
$n\in \overline{\mathcal{U}} ^ { (i) }$,
we first form a candidate likelihood using all other bits' information in the same
word:
\begin{equation}
\phi_{n,\backslash k}^{(i)}(\mathbf{w})= \prod_{\substack{\ell \in \mathcal{B}_n\\ \ell \ne k}}
q^{(i)}\big(v_\ell=b_\ell(\mathbf{w})\big),\; \mathbf{w}\in\mathcal{C}_n.
\label{eq:16}
\end{equation}
The exclusion of the target bit
$k$
ensures that the evidence carried by
$q ^ { (i) } (v_{ k } )$
is not fed back into the prior of the same bit, so that the same
evidence is not counted twice. The word-level posterior used to update
bit
$k$
is obtained by combining the MLM likelihood and the BCJR likelihood:
\begin{equation}
\rho_{n,\backslash k} ^{ (i) }\left ( { \mathbf{w} } \right )=\mathsf{N}_{ \mathbf{w},\mathcal{C}_{ n } } \Big ( { \theta_{ n } ^{ (i) }(\mathbf{w})\cdot\phi_{n,\backslash k} ^{ (i) }(\mathbf{w}) } \Big ),\;\mathbf{w}\in\mathcal{C}_{ n } .
\label{eq:17}
\end{equation}
Finally, the updated CLP for the bit
$k$
covered by a masked word is obtained by marginalizing over all candidate
words whose bit value at position
$k$
equals
$v\in\mathbb{F}_{ 2 }$:
\begin{align}
&p ^ { (i+1) } \left ( { v_{ k } =v } \right )
\nonumber
\\
&=\sum \limits_{ \mathbf{w}\in\mathcal{C}_{ n } } { \rho_{n,\backslash k} ^{ (i) }\left ( { \mathbf{w} } \right )\mathbb{I}\left ( { b_{ k } (\mathbf{w})=v } \right ) },\; k\in\mathcal{B}_{ n } ,n\in \overline{\mathcal{U}} ^ { (i) } ,
\label{eq:18}
\end{align}
where
$\mathbb{I}(\cdot)$
is the indicator function. The resulting natural-order priors
$p ^ { (i+1) } (v_{ k } )$
are then interleaved by
$\pi$
to obtain
$p ^ { (i+1) } (u_{ k } )$,
which serves as the bit prior for the BCJR decoding in the next
iteration.

\begin{remark}
The reader may wonder why
$q ^ { (i) } \left ( { v_{ \ell } } \right ),\,\ell \ne k$,
in \eqref{eq:16}, 
can be treated as extrinsic information for bit
$k$ by the BCJR decoder,
even though
$q ^ { (i) } \left ( { v_{ \ell } } \right )$
are actually intrinsic information produced by the BCJR decoder. The
rationale is as follows. Because of the interleaver, the bits associated
with the same word are not necessarily adjacent in the convolutional
codeword.

\begin{itemize}
  \item \textbf{Exact vs. Approximate:} If the interleaver is designed so that any two bits belonging to the same word are separated by at least $C$ bit positions, where $C$ is the constraint length of the convolutional code, then $\phi_{n,\backslash k} ^{ (i) }\left ( { \mathbf{w} } \right )$ in \eqref{eq:16} is strictly extrinsic information for bit $k$. This is because, under such an interleaving constraint, the bits of the same word do not share local dependencies induced by the convolutional-code trellis. Otherwise, treating it as such is an approximation. 
  \item \textbf{Design Choice:} To maintain simplicity, simulations in Section~\ref{simulation-results} use an unrestricted interleaver design. The strong performance achieved under this setting demonstrates that our approach is robust to the resulting approximation. This suggests that our approach can also be applied to LDPC codes, for which there is no direct notion of a constraint length. Algorithm details and the associated simulation results of the LDPC-based ICL-SEC are given in Appendix~\ref{app:ldpc}.
\end{itemize}

\end{remark}

\begin{remark}
The reader is reminded that the computation of
$p ^ { (i+1) } \left ( { v_{ k } } \right )$
in \eqref{eq:16}--\eqref{eq:18} applies to masks only. In
the mask-unmask word sequence
$\hat {\mathbf{s}}_{ \mathrm{m} } ^{ (i) }$
presented to the MLM, no information is given by the BCJR decoder on the
masked words. Specifically, the information contained in
$q ^ { (i) } \left ( { v_{ \ell } } \right )$ with $\ell \ne k$,
in \eqref{eq:16}, has not been used twice by the semantic
decoder consisting of the MLM and the CLP Updater.
\end{remark}

\vspace{1mm}
\noindent\textbf{2) Properties of Confirm Scheme}
\vspace{1mm}


\begin{proposition}[Non-Oscillatory Confirmation]\label{prop1}
In the Confirm scheme, once a word becomes unmasked and is confirmed at iteration
$i$,
it remains unmasked in all subsequent iterations, and all its bits
retain the same values. Specifically, for any confirmed word
$\hat {\mathbf{w}}_{ n } ^{ (i) }$ at
iteration \emph{i} with bit values
$(\hat{v}_k^{(i)})_{k\in\mathcal{B}_n}$,
we have
$\hat {\mathbf{w}}_{ n } ^{ (j) }= \hat {\mathbf{w}}_{ n } ^{ (i) }$,
$\hat{v}_k^{(j)}=\hat{v}_k^{(i)}$,
and
$\mu_{ n } ^{ (j) }( \hat {\mathbf{w}}_{ n } ^{ (j) })=1$ for all iterations $j > i$.
Thus, a confirmed word cannot oscillate between the unmasked and masked
states.
\end{proposition}

\begin{IEEEproof}
By induction, it suffices to prove the proposition for
$j=i+1$.
Consider a particular confirmed word $n\in\mathcal{U}^{(i)}$ at iteration $i$. For
iteration
$i+1$,
let
$\gamma_{\pi(k)}^{(i+1)}(s',s)$
be the BCJR branch metric at the natural-order bit 
$k\in\mathcal{B}_n$
with an associated input bit
$u_{\pi(k)}=v_k$.
With reference to \eqref{eq:8}, the input-bit prior is a
multiplicative factor in the branch metric, and thus every branch
inconsistent with the HD value has zero metric
$\gamma_{\pi(k)}^{(i+1)}(s',s)=0$,
whenever the input bit
$u_{\pi(k)}=v_k\ne\hat{v}_k^{(i)},\, k\in\mathcal{B}_n$.
According to \eqref{eq:5} and \eqref{eq:1}, for every
$k\in\mathcal{B}_n$,
we therefore have
\begin{equation}
r^{(i+1)}(v_k=v)=\begin{cases}1,&v=\hat{v}_k^{(i)},\\0,&v\ne\hat{v}_k^{(i)},\end{cases}\;\;v\in\mathbb{F}_2.
\label{eq:19}
\end{equation}

Substituting these deterministic APPs into the word-candidate scoring in
\eqref{eq:11}, the candidate
$\hat {\mathbf{w}}_{ n } ^{ (i) }$
is the only one whose bit-wise probability factors are all non-zero; 
every other candidate in
$\mathcal{C}_{ n }$
differs from
$\hat {\mathbf{w}}_{ n } ^{ (i) }$
in at least one bit and hence receives zero score. Therefore, we have
\begin{equation}
\mu_n^{(i+1)}(\mathbf{w})=\begin{cases}1,&\mathbf{w}=\hat{\mathbf{w}}_n^{(i)},\\0,&\mathbf{w}\ne\hat{\mathbf{w}}_n^{(i)},\end{cases}\;\; \mathbf{w}\in\mathcal{C}_n.
\label{eq:20}
\end{equation}
With \eqref{eq:12}, we arrive at
$\hat {\mathbf{w}}_{ n } ^{ (i+1) }= \hat {\mathbf{w}}_{ n } ^{ (i) }$ and $\Xi_{ n } ^{ (i+1) }=1\ge T$.
Thus, by the HD-and-masking rule, word
$n$
remains unmasked in all subsequent iterations, and all its covered bits
remain the same values.
\end{IEEEproof}

\begin{corollary}[Non-Decreasing Confirmation]
In the Confirm
scheme, the confirmed-word set is non-decreasing, i.e.,
\begin{equation}
\mathcal{U}_{ i-1 } \subseteq \mathcal{U}_{ i } , \;\,i=1,\ldots,I.
\label{eq:21}
\end{equation}
\end{corollary}

\begin{IEEEproof}
Implied by Proposition~\ref{prop1}.
\end{IEEEproof}

\begin{remark}[Consistent Reliability Interpretation]\label{re:consisent}
In the
Confirm scheme, an unmasked word has the same reliability meaning for
both the BCJR decoder and the MLM-empowered semantic decoder. Once the
HD-and-masking component in Fig.~\ref{fig:iclsec-detail} declares a word to be unmasked, the
word is treated as fully reliable by the semantic decoder as well as the
BCJR decoder in future iterations.
\end{remark}

\begin{remark}[Decoder-Agnostic Guarantee]\label{re:DecoderAgnosticGuarantee}
The proof of Proposition \ref{prop1} uses only the fact that the prior probability enters the APP multiplicatively, as in (\ref{eq:9}). Thus, the non-oscillatory guarantee holds for any soft-input soft-output decoder, including the LDPC BP decoder in Appendix \ref{app:ldpc}, whose source-bit APP is the product of the prior probability and the incoming beliefs from neighboring factor nodes.
\end{remark}

\begin{remark}[Analogy with Successive Interference Cancellation]\label{re:analogy}
The successive-confirmation mechanism is analogous to successive
interference cancellation (SIC) decoding for overlapped signals in
wireless networks. In SIC, once a constituent signal within a
superposition of signals is decoded with sufficiently high confidence,
it is fixed and subtracted from the received signal, facilitating the
decoding of the remaining signals. As more signals are successively
decoded and canceled, fewer unresolved signals remain, until eventually
all signals are decoded. Similarly, in our successive confirmation, once
certain words are deemed reliable, they are fixed and used as semantic
context to help recover the remaining masked words. As confirmations
proceed, the number of unmasked words increases, until all words are
unmasked in the final iteration.
\end{remark}

\subsection{Naive and Hold CLP Update}\label{naive-and-hold-clp-update}

We next introduce two additional CLP update schemes, referred to as
\textbf{Naive} and \textbf{Hold}. Our investigation in this paper
originally began with the Naive scheme. The Hold scheme was subsequently
devised to mitigate the oscillatory behavior observed in Naive for all
text messages. Although Hold removes this behavior in a broad sense,
oscillations may still occur for specific text messages. The Confirm
scheme was ultimately developed as a consistent iterative scheme with
provably non-oscillatory behavior.

Confirm, Naive, and Hold schemes, share the same components in Fig.~\ref{fig:iclsec-detail} except the
CLP updater module. Therefore, we focus the Naive and Hold CLP updater
here.

Before proceeding, we note that the MLM, e.g., mmBERT \cite{Marone2025mmBERT}, produces a
probability distribution
$\theta_{ n } ^{ (i) }(\mathbf{w})$ over
candidate words at every word position, including both masked and unmasked
positions. The Confirm scheme, however, does not use the distribution at
unmasked positions. The Naive scheme, in contrast, also makes use of the
unmasked-position probability distribution, while the Hold scheme does not exploit such information like Confirm scheme. 

\textbf{Naive Scheme:} In this scheme, the MLM-empowered CLP updater
refines the priors of all words at every iteration, regardless of
whether they are masked or unmasked. Specifically, for each word
$n\in\mathcal{N}$
and each its covered bit
$k\in\mathcal{B}_{ n }$,
the prior is updated using the same marginalization rule as in
\eqref{eq:18}:
\begin{align}
&p ^ { (i+1) } \left ( { v_{ k } =v } \right )
\nonumber \\
&=\sum \limits_{ \mathbf{w}\in\mathcal{C}_{ n } } { \rho_{n,\backslash k} ^{ (i) }\left ( { \mathbf{w} } \right )\mathbb{I}\left ( { b_{ k } (\mathbf{w})=v } \right ) }, \;k\in\mathcal{B}_{ n } ,n\in\mathcal{N},
\label{eq:22}
\end{align}
where
$\rho_{n,\backslash k}^{(i)}(\mathbf{w})$
used in Naive is defined in the same way as in
\eqref{eq:17}. The key difference from Confirm is
that the Naive scheme also updates the priors of unmasked words instead
of fixing the associated bit priors as probability one.

This design leads to a key shortcoming. Recall that we combine the MLM
likelihood and the BCJR likelihood in \eqref{eq:17}. For a
masked word
$n$,
the MLM likelihood is inferred from the surrounding context of the word
(i.e., words surrounding word
$n$).
However, for an unmasked word $n$, the MLM can observe the decoded word
$\hat {\mathbf{w}}_{ n } ^{ (i) }$
itself. As a result, the MLM likelihood is no longer inferred only from
surrounding context; it also reuses information from word
$n$.
In \eqref{eq:17}, the information from the unmasked word
$n$
also contributes to
$\phi_{n,\backslash k}^{(i)}(\mathbf{w})$.
Therefore, the prior update in Naive creates a feedback loop
and double-counts the same evidence within each unmasked word.

This loop can perturb words that were already deemed reliable.
Consequently, an unmasked word may become masked again in a later
iteration, leading to oscillation and instability (see the results in
Section~\ref{iteration-performance}). In particular, the Naive scheme does not satisfy
Proposition~\ref{prop1}. Neither does it have the consistent reliability
interpretation mentioned in Remark~\ref{re:consisent}. Since an unmasked word is not
interpreted using bit priors with probability one, the reliability
interpretation is inconsistent between the MLM and BCJR: the MLM sees an
unmasked word as truly trusted hard context, whereas BCJR still receives
only soft priors.

\begin{remark}
We refer to this scheme as Naive because it applies
the sum-product BP principle in a broad but overly
simplistic manner. Specifically, the sum-product operation in
\eqref{eq:22} is applied to all word positions without
accounting for the particular structure of the scheme in Fig.~\ref{fig:iclsec-detail}. In
particular, the input to the MLM is a sequence of unmasked words and
masks, rather than a probability distribution over candidate words at
each position. This differs from a standard BP network,
in which all beliefs propagated between nodes are probabilistic beliefs.
\end{remark}

\textbf{Hold Scheme:} Unlike the Naive scheme, the Hold scheme does not
refine the priors of unmasked words for the next iteration.
Specifically, if word
$n\in\mathcal{U} ^ { (i) }$
is unmasked at iteration
$i$,
then for each
$k\in\mathcal{B}_{ n }$,
the prior is held:
\begin{equation}
p ^ { (i+1) }\left ( { v_{ k } =v } \right )=p ^ { (i) } \left ( { v_{ k } =v } \right ),\;v\in\mathbb{F}_{ 2 } ,k\in\mathcal{B}_{ n } , n\in\mathcal{U} ^ { (i) } .
\label{eq:23}
\end{equation}
For masked words, the Hold scheme uses the same CLP-update procedure as
in \eqref{eq:16}--\eqref{eq:18}.

The Hold scheme avoids the direct double-counting issue of Naive
because it does not directly refine the priors for the unmasked
words. However, holding a soft prior is still fundamentally different
from confirming a word with probability one. An unmasked word in the
Hold scheme may become masked in a later iteration even though its prior
was held fixed after the iteration in which it was unmasked. This is
because the BCJR APP of that word depends not only on its own prior, but
also on the forward and backward beliefs propagated through the
trellis. When the priors of other words change, the information entering
the current word also changes. Therefore, a previously unmasked word can fail
the masking threshold later.

Once such a word becomes masked again, its prior is no longer held and
is renewed again by the CLP updater in \eqref{eq:16}--\eqref{eq:18}. Therefore, the Hold scheme may
still exhibit unmask-to-mask transitions across iterations, which is
verified by our simulations in Section~\ref{iteration-performance}. In other words, it does not
provide the non-oscillatory guarantee in Proposition~\ref{prop1}. The Hold scheme
also does not satisfy the consistency property mentioned in Remark~\ref{re:consisent}.
When a word is unmasked, the MLM observes it as a hard, fully trusted
word. However, BCJR only receives the previously held soft prior rather
than a deterministic prior with probability one. Therefore, the same
unmasked word has different reliability meanings in the two modules.

\section{Simulation Results}\label{simulation-results}

\subsection{Experimental Setups}\label{experimental-setups}

\textbf{Model and Dataset:} We employ mmBERT
\cite{Marone2025mmBERT} as the MLM for the CLP update. We use the
50 English passages \cite{Wang2026CLSEC} as
the transmitted messages. For each tested SNR, the transmission of each
passage is simulated over 100 Monte-Carlo trials with independently
generated channel noise, yielding 5000 received corrupted passages.

\textbf{System Settings:} We instantiate the physical-layer channel code
using a convolutional code with generator polynomial $(1,5/7)_{\mathrm{oct}} $, rate
$R=1/2$, 
and constraint length $C=3$. 
We apply a random interleaver $\pi$. For the iterative schemes, the maximum
number of iterations is set to
$I_{ \mathrm{max} } =10$,
and the prior-difference stopping threshold is set to
$\epsilon=10 ^ { -3 }$.
For the masking probability threshold
$T$,
we adopt an operating-point-dependent setting. Specifically,
$T$
is calibrated offline as a function of the operating SNR and stored in a
lookup table, analogous to the standard adaptive modulation and coding
(AMC) practice in modern receivers, where transmission parameters are
selected according to the estimated channel quality. This incurs no
additional online cost, since SNR estimation is routinely available at
the receiver.

\textbf{Investigated Schemes:} We compare five schemes. The first two
are non-iterative, while the last three are iterative and differ only in
the CLP prior-update rule. Let
$\{ \hat {\mathbf{w}}_{ n } ^{ (i) }\}_{ n\in\mathcal{N} }$
be the decoded words and
$\{ \hat v_{ k } ^{ (i) }\}_{ k\in\mathcal{K} }$
be the associated HD bits across all schemes, where the iteration index
$i=1$
for the non-iterative schemes and
$i=1,\ldots,I\le I_{\max}$
for the iterative schemes. The determination of
$\{ \hat {\mathbf{w}}_{ n } ^{ (i) }\}_{ n\in\mathcal{N} }$
and
$\{ \hat v_{ k } ^{ (i) }\}_{ k\in\mathcal{K} }$
in iterative schemes has been elaborated in Section~\ref{bcjr-intrinsic-information-hard-decision-and-masking}.

\textbf{1) BCJR Decoding:} This is the conventional channel decoding
baseline. The BCJR decoding is run once with uniform bit priors to
obtain bit-level APPs
$r ^ { (1) } (v_{ k } ),\,k\in\mathcal{K}$.
Then the conventional bit-wise HD is performed:
\begin{equation}
\hat v_{ k } ^{ (1) }=\operatorname*{arg\,max} \limits_{ v\in\mathbb{F}_{ 2 } }\, r ^ { (1) } \left ( { v_{ k } =v } \right ).
\label{eq:24}
\end{equation}
Converting
$\{\hat{v}_k^{(1)}\}_{k\in\mathcal{B}_n}$
for each
$n\in\mathcal{N}$
to characters yields decoded words
$\{ \hat {\mathbf{w}}_{ n } ^{ (1) }\}_{ n\in\mathcal{N} }$.
This pure physical-layer scheme does not use any semantic information.
The bit sequence in
$\hat {\mathbf{w}}_{ n } ^{ (1) }$
may not even be a valid word in the vocabulary.

\textbf{2) CL-SEC Scheme:} This is the non-iterative baseline from
\cite{Wang2026CLSEC}. The BCJR decoding is
first run once with uniform priors to obtain bit-level APPs
$r ^ { (1) } (v_{ k } ),\,k\in\mathcal{K}$.
Then the word-wise HD and masking in Section~\ref{bcjr-intrinsic-information-hard-decision-and-masking} are performed,
producing
$\{ \hat {\mathbf{w}}_{ n } ^{ (1) }\}_{ n\in\mathcal{N} }$
and
$\overline{\mathcal{U}} ^ { (1) }$.
To replace each unreliable decoded word that is masked, CL-SEC
combines the physical-layer word likelihood obtained from BCJR
probabilities with the MLM likelihood acquired from semantic context.
That is,
\begin{align}
\label{eq:25}
  \delta_n^{(1)}(\mathbf{w})&=\prod_{k\in\mathcal{B}_n}r^{(1)}(v_k=b_k(\mathbf{w}))\cdot\theta_n^{(1)}(\mathbf{w}),\; n\in\overline{\mathcal{U}}^{(1)}, \\ 
\hat{\mathbf{w}}_n^{(1)}&\leftarrow\underset{\mathbf{w}\in\mathcal{C}_n}{\arg\max}\ \delta_n^{(1)}(\mathbf{w}),\; n\in\overline{\mathcal{U}}^{(1)}.
\end{align}
Mapping $\{ \hat {\mathbf{w}}_{ n } ^{ (1) }\}_{ n\in\mathcal{N} }$
to ASCII bits
yields $\{ \hat v_{ k } ^{ (1) }\}_{ k\in\mathcal{K} }$.
Unlike ICL-SEC, this scheme performs the cross-layer correction only
once and does not generate prior to feed back to BCJR decoder.

\textbf{3) Naive Scheme:} This scheme has been detailed in Section~\ref{naive-and-hold-clp-update}.
At each iteration, it updates the priors of both masked and unmasked
words using the MLM-empowered CLP updater, as in
\eqref{eq:22}.

\textbf{4) Hold Scheme:} This scheme has also been introduced in Section~\ref{naive-and-hold-clp-update}. For words that are unmasked at an iteration, the scheme keeps
their priors unchanged in the next iteration (see
\eqref{eq:23}); for masked words, it updates the priors
using the MLM-empowered CLP updater (see \eqref{eq:16}--\eqref{eq:18}).

\textbf{5) Confirm Scheme (Proposed):} The Confirm scheme is the proposed
ICL-SEC scheme with successive confirmation (see Section~\ref{confirm-clp-update-ultimate-proposed-scheme}). Once a
word is unmasked with sufficient reliability, it is confirmed and its
corresponding bit priors are fixed to probability one in all
subsequent iterations, as in \eqref{eq:15}. Bit priors for
the masked words are updated using the MLM-empowered CLP updater, as
in \eqref{eq:16}--\eqref{eq:18}.

\textbf{Evaluation Metrics:} We evaluate both transmission reliability
and semantic fidelity. For transmission reliability, we report bit error
rate (BER), word error rate (WER), and character error rate (CER). For
semantic fidelity, we report BERTScore \cite{Zhang*2020BERTScore},
BLEU \cite{Papineni2001BLEU}, and ROUGE-L \cite{Lin2004ROUGE}.
BERTScore measures semantic similarity through contextual token
embeddings. BLEU measures
$n$-gram
overlap between the recovered and true texts. ROUGE-L measures
sequence-level overlap based on the longest common subsequence.

The performance metrics are computed as follows for each
$i$. 
First, BER is computed by comparing the bit stream
$( \hat v_{ k } ^{ (i) })_{ k\in\mathcal{K} }$
with its ground truth
$(v_{ k })_{ k\in\mathcal{K} }$.
Second, we obtain WER by comparing
$\{ \hat {\mathbf{w}}_{ n } ^{ (i) }\}_{ n\in\mathcal{N} }$
with
$\{{\mathbf{w}}_n\}_{ n\in\mathcal{N} }$;
a word is counted as correct only if all its letters are recovered
correctly. Third, CER is evaluated by comparing the character string
$( \hat {\mathbf{w}}_{ n } ^{ (i) })_{ n\in\mathcal{N} }$
with
$( {\mathbf{w}}_n)_{ n\in\mathcal{N} }$.
Finally, we obtain the recovered spaced text 
$\hat {\mathbf{s}} ^ { (i) }$
by inserting spaces between adjacent words in
$\{ \hat {\mathbf{w}}_{ n } ^{ (i) }\}_{ n\in\mathcal{N} }$, and input
$\hat {\mathbf{s}} ^ { (i) }$
and its ground-truth counterpart to each semantic scorer for the
semantic-fidelity computation. Since the BCJR scheme does not use any
semantic information, its severely poor semantic-fidelity results are
omitted in the figures for clarity.

\begin{figure*}[!t]
\centering
\includegraphics[width=0.85\linewidth]{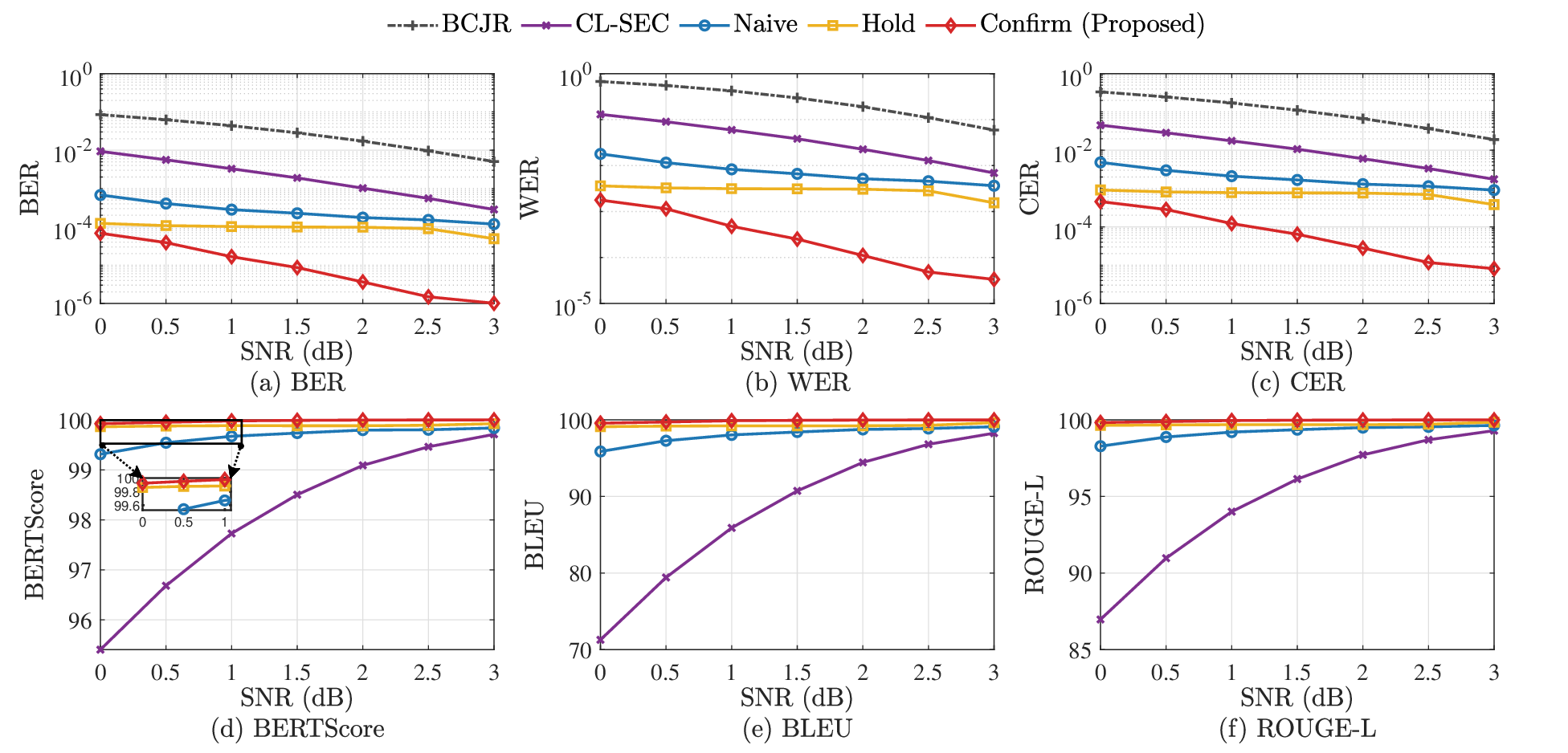}
\caption{Performance comparison with varying SNR. The first row reports (a) BER, (b) WER, and (c) CER; the second row reports (d) BERTScore, (e) BLEU, and (f) ROUGE-L.}
\label{fig:snr-performance}
\end{figure*}

\begin{figure*}[!t]
\centering
\includegraphics[width=0.85\linewidth]{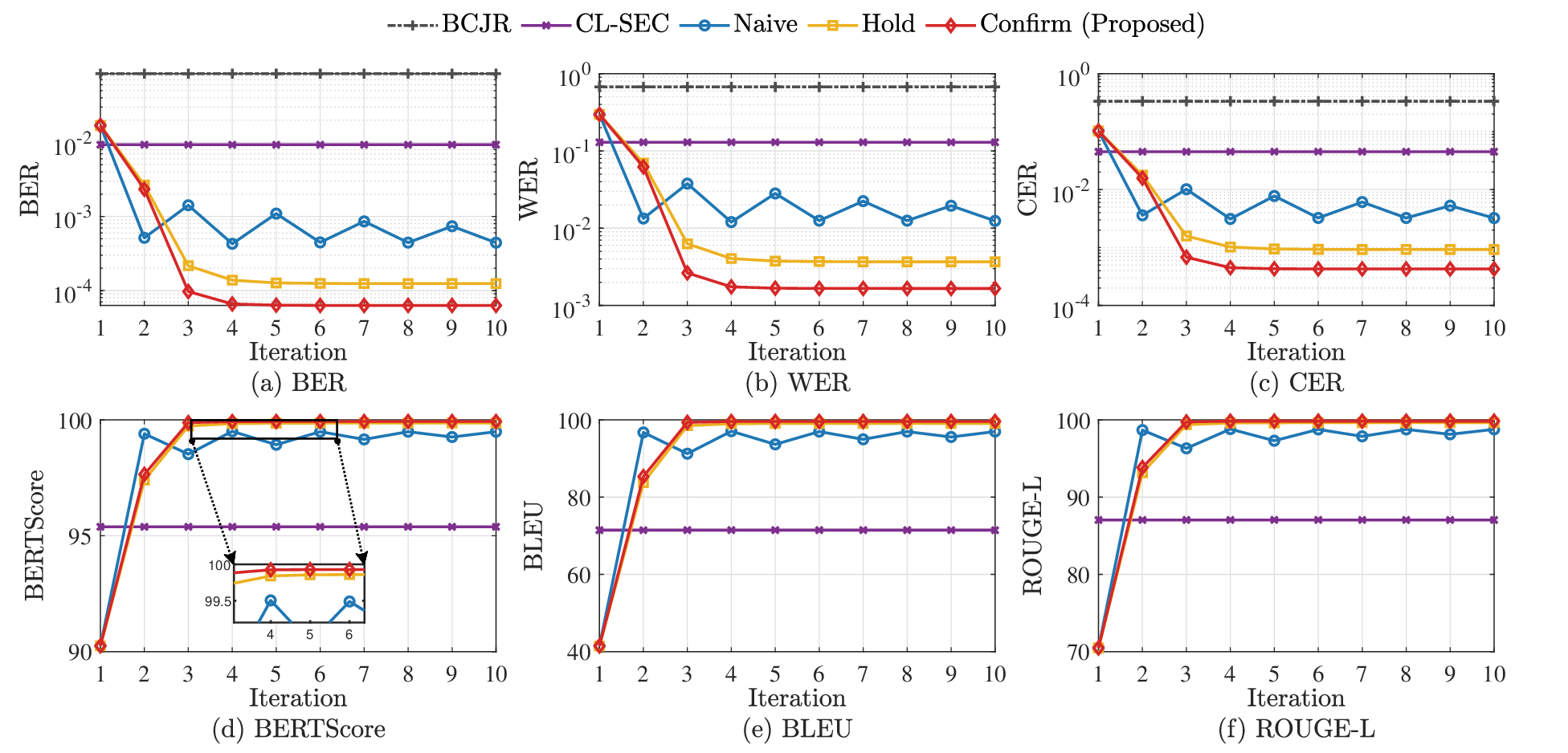}
\caption{Per-iteration performance at 0\,dB SNR, averaged over 5000 Monte-Carlo trials. The first row reports (a) BER, (b) WER, and (c) CER; the second row reports (d) BERTScore, (e) BLEU, and (f) ROUGE-L.}
\label{fig:iteration-performance}
\end{figure*}\label{fig:transitions}

\subsection{Varying-SNR Performance}\label{varying-snr-performance}

Fig.~\ref{fig:snr-performance} compares the performance of the investigated
schemes over different SNRs. The results of iterative schemes are
collected from the final iteration. Across all six metrics, the Confirm
ICL-SEC scheme consistently achieves the best performance. Importantly,
Confirm ICL-SEC achieves a substantial improvement over the
non-iterative CL-SEC. For example, at 1.5\,dB SNR, the BER of
non-iterative CL-SEC is
$1.90\times10 ^ { -3 }$,
while Confirm ICL-SEC reduces it to
$8.64\times10 ^ { -6 }$,
constituting a BER reduction by a factor of
$220$;
across all considered SNRs, Confirm ICL-SEC reduces the BER by more than
two orders of magnitude compared with CL-SEC. Moreover, the gain over
the conventional BCJR is larger: BCJR has a BER of
$2.85\times10 ^ { -2 }$
at 1.5\,dB SNR, and thus Confirm scheme achieves a BER reduction by a
factor of
$3299$,
corresponding to more than three orders of magnitude.

The Confirm scheme also outperforms the other iterative approaches. For
instance, compared with the Naive and Hold schemes, Confirm
reduces the WER at 1.5\,dB SNR by factors of
$26.3$
and
$12.4$,
respectively. This indicates that iteration alone is not sufficient --
the way in which the CLP is updated is critical. Both the Naive and Hold
schemes may render the unmasked word become masked again in a later
iteration, as shown in Fig.~\ref{fig:iteration-performance} and
Fig.~\ref{fig:unmask-to-mask}. This limitation prevents them from achieving
the superior performance of the proposed successive confirmation design.

From Fig.~\ref{fig:snr-performance}, the same trend is also observed in CER, which is consistent with the BER
and WER results. Semantic-fidelity results further confirm the advantage
of Confirm ICL-SEC. At 0\,dB SNR, Confirm scheme achieves a BERTScore,
BLEU, and ROUGE-L of 99.92, 99.57, and 99.82, respectively,
demonstrating the accurate semantic recovery even at low SNR.

\subsection{Iterative Performance}\label{iteration-performance}

Fig.~\ref{fig:iteration-performance} examines the evolution of different schemes
over iterations at 0\,dB SNR, averaged over 5000 Monte-Carlo trials. Since
BCJR and CL-SEC schemes are non-iterative, they appear as horizontal
reference lines. Confirm ICL-SEC achieves the best final performance and
improves monotonically without oscillation. Moreover, ICL-SEC exhibits a
fast convergence rate. For 0\,dB SNR, it merely requires about five
iterations on average to converge.

In contrast, the Naive scheme exhibits severe oscillatory behavior: its
error metrics can decrease in one iteration and then increase again in a
later iteration. The oscillations are persistent and uniform over the
5000 trials and manifest themselves in the results of all metrics despite the
averaging over the 5000 trials. From Fig.~\ref{fig:iteration-performance}, we see that even iterations
having better performance (e.g., BER) than odd iterations. Furthermore, the BER performance
does not improve further after iteration 4. 

The Hold scheme appears more stable after averaging over many trials,
but it still lacks the non-oscillatory guarantee of the Confirm scheme
as mentioned. From our experiments, we observe that the oscillations of
Hold vary from trial to trial and is trial-specific (see Fig.~\ref{fig:unmask-to-mask}).

\begin{figure}[!t]
\centering
\includegraphics[width=1\linewidth]{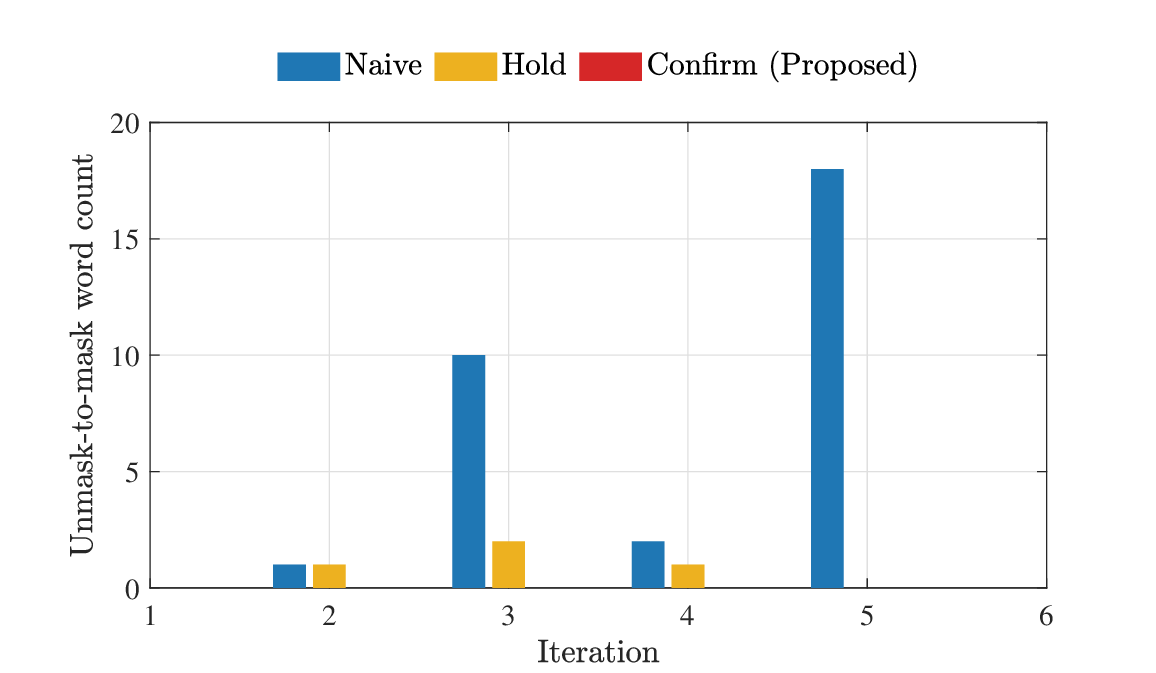}
\caption{Unmask-to-mask transitions for a single trace operated at 0\,dB SNR with $N=105$ words. }
\label{fig:unmask-to-mask}
\end{figure}

Fig.~\ref{fig:unmask-to-mask} provides a closer look at the unmask-to-mask
transitions for a single trace at 0\,dB SNR. Specifically, at iteration
$i\ge2$,
we count the number of unmasked words at iteration
$i-1$
which revert to masked words at iteration
$i$.
The transmitted message contains
$N=105$
words. Notably, the Confirm scheme has zero unmask-to-mask transitions.
This result directly validates Proposition~\ref{prop1}: once a word is confirmed,
its associated bit priors are fixed to probability one, and hence the word
cannot become masked again in later iterations. In contrast, the Naive
scheme exhibits frequent decision reversals, with 31 unmask-to-mask
transitions in this trace. The Hold scheme reduces but does not
eliminate such transitions, with 4 unmask-to-mask events.

\section{Conclusion}\label{conclusion}
This paper proposed \textbf{ICL-SEC}, an iterative cross-layer semantic error correction framework that integrates physical-layer soft information with application-layer language-model-driven contextual inference through repeated feedback. Motivated by the success of turbo and LDPC decoding, the proposed approach treats semantic recovery not as a one-shot post-processing task, but as an iterative refinement process in which the channel decoder and the semantic decoder exchange complementary reliability information.

The central principle of ICL-SEC is \textbf{successive confirmation}. After each soft-input soft-output channel-decoding pass, word-level posterior probabilities are used to identify reliable words. These words are confirmed and then revealed to the masked language model as trusted context, while unreliable words remain masked. The language model generates contextual word likelihoods for the
unresolved positions, and these likelihoods are leveraged to refine the bit-level priors for the next channel-decoding iteration. Unlike Naive iterative updates that repeatedly revise all words, the proposed Confirm scheme assigns deterministic priors with probability one to words once they are declared reliable. This design gives an unmasked word the same meaning in both layers: it is fully trusted semantic context for the application-layer language model and fully reliable side information for the physical-layer channel decoder.

This reliability-consistent design leads to an important structural
property: confirmed words cannot revert to masked words in later
iterations. The resulting non-oscillatory behavior distinguishes the
proposed method from the \textbf{Naive} and \textbf{Hold} variants, both
of which can suffer from instability because previously unmasked words
may later become unreliable again. The \textbf{Confirm} scheme therefore
does more than add iteration; it provides a principled prior-update rule
that stabilizes the interaction between physical-layer decoding and
semantic inference.

Simulation results demonstrate the significance of this design. Across
BER, WER, CER, BERTScore, BLEU, and ROUGE-L, ICL-SEC consistently
improves over conventional BCJR decoding, non-iterative CL-SEC, and the
alternative iterative update approaches. The proposed Confirm scheme reduces
BER by more than two orders of magnitude compared with non-iterative
CL-SEC and achieves high semantic fidelity even at low SNR. Moreover, since ICL-SEC is agnostic to the underlying channel code and requires only a soft-input soft-output decoder, the same framework applies beyond convolutional codes: an LDPC instantiation under belief-propagation decoding reduces BER by more than four orders of magnitude relative to conventional belief-propagation decoding, confirming the generality of the proposed approach. These gains
show that semantic information, when fed back in a carefully controlled
and reliability-aware manner, can substantially enhance communication
reliability.

More broadly, this work suggests a new direction for semantic
communications: semantic intelligence should not merely repair corrupted
outputs after channel decoding, but should participate actively in the
decoding loop. By aligning channel reliability and semantic confidence
through iterative cross-layer exchange, ICL-SEC provides a step toward
communication systems that are both error-resilient and meaning-aware.
Future work may extend this principle to larger
language models, adaptive threshold design, and non-textual modalities
such as images, speech, and multimodal messages.

\begin{figure*}[!t]
\centering
\includegraphics[width=0.84\linewidth]{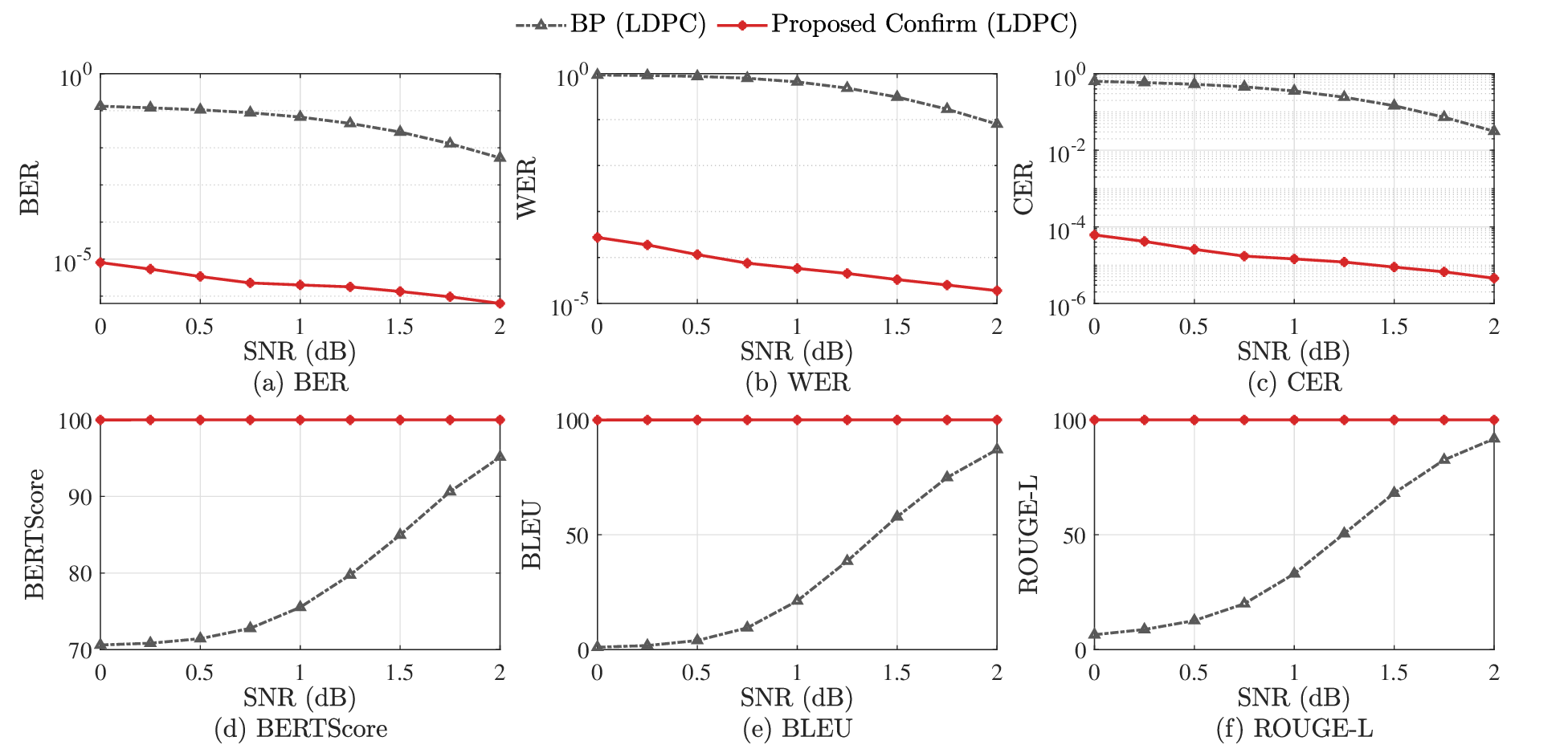}
\caption{LDPC-based performance comparison. The first row reports (a) BER, (b) WER, and (c) CER; the second row reports (d) BERTScore, (e) BLEU, and (f) ROUGE-L.}
\label{fig:snr-performance-ldpc}
\end{figure*}

\appendices
\section{LDPC-Based Instantiation of ICL-SEC}
\label{app:ldpc}

The main sections instantiate ICL-SEC using a convolutional code and BCJR decoding. This appendix describes how the same framework can be instantiated using an \textbf{LDPC code}. The only modification is to replace the convolutional encoder and BCJR decoder with an LDPC encoder and a soft-input soft-output belief-propagation (BP) decoder, respectively. 

\textbf{Specification:} 
At the transmitter, the bit stream $\mathbf{u}$ is encoded by a rate-$R$ systematic LDPC code, yielding the coded-bit sequence $ \mathbf{c} $. The coded bits are then modulated and transmitted over the channel, as described in Section~\ref{system-architecture}.

The receiver follows the same three-step ICL-SEC procedure presented in Section~\ref{overview}. Only the channel-decoding operation in \textbf{Step~1} is modified. Accordingly, the upper convolutional-code factor graph in Fig.~\ref{fig:iclsec-detail} is replaced by the factor graph associated with the LDPC code. The HD-and-masking module, MLM component, CLP updater, and interleaver/deinterleaver in the lower part of Fig.~\ref{fig:iclsec-detail} remain unchanged. Note that only the systematic codeword positions, which carry $u_k$'s, are connected to the lower modules. Accordingly, the CLP updater refines only the priors of the systematic bits, while the parity-bit priors remain uniform.

Specifically, in Step 1, perform BP decoding with the prior $p^{(i)}(u_k)$ to produce the complete APP $r^{(i)}(u_k)$ and BP-intrinsic information $q^{(i)}(u_k)$, satisfying \eqref{eq:9}; and then employ the same HD-and-masking module in Fig.~\ref{fig:iclsec-detail} to examine the most likely word at each word position. 
In Step~2 and Step~3, perform MLM-driven contextual inference and CLP update, respectively, using the same procedure elaborated in Section~\ref{overview}. Repeat the above procedure until the maximum iteration count is reached or an early stopping rule is triggered. 

\textbf{Simulation Results:} 
Fig.~\ref{fig:snr-performance-ldpc} presents the performance of the LDPC-based schemes. We employ a (3,6)-regular, systematic $(720,362)$ LDPC code generated using the pyldpc library \cite{pyldpc}. Each payload is partitioned into source blocks of 362 information bits; the final block is zero-padded when necessary before LDPC encoding. The resulting LDPC codewords are concatenated and rate matched over the entire payload, operating at an effective code rate of $1/2$. The BP baseline in Fig.~\ref{fig:snr-performance-ldpc} corresponds to the pure physical-layer decoding approach (Scheme~1)  in Section~\ref{experimental-setups}, with the convolutional code and BCJR decoder replaced by the LDPC code and BP decoder, respectively.
The maximum number of iterations in LDPC BP decoding is set to 80. 

We see from Fig.~\ref{fig:snr-performance-ldpc} that our LDPC-based Confirm ICL-SEC achieves a significant gain over conventional BP decoder. 
At 1\,dB SNR, BP decoder achieves a BER of $6.75\times10^{-2}$, while Confirm scheme reduces the BER to $2.01\times10^{-6}$. This corresponds to a $3.36\times10^{4}$-fold reduction, an improvement of more than four orders of magnitude. Across all evaluated SNRs, Confirm provides an average BER reduction of approximately $2.12\times10^{4}$ over BP decoding. The WER and CER results exhibit the same overall trend. 
Moreover, the semantic-fidelity results further demonstrate the benefit of the Confirm scheme. For example, at 0\,dB SNR, Confirm achieves a BERTScore of $99.99$, a BLEU score of $99.94$, and a ROUGE-L score of $99.98$, significantly outperforming the conventional BP decoding. 

These results demonstrate that ICL-SEC is a general framework that can operate with different channel codes decoded by soft-input soft-output decoders. Under the evaluated settings, the LDPC-based instantiation achieves even lower error rates than its convolutional-code-based counterpart, confirming the benefits of Confirm ICL-SEC built on stronger channel codes. 

\bibliographystyle{IEEEtran}
\bibliography{ICL-SEC_refs}

\end{document}